\documentclass{llncs}

\usepackage[utf8]{inputenc}
\usepackage[T1]{fontenc}
\usepackage{amsmath,amssymb}
\usepackage{graphicx}
\usepackage{hyperref}
\usepackage{paralist}
\usepackage{todonotes}
\usepackage{xspace}
\usepackage{booktabs}
\usepackage{multirow}
\usepackage{caption}
\usepackage{subcaption}
\usepackage{times}
\usepackage[pagewise]{lineno}
\usepackage{xcolor}
\definecolor{blue}{rgb}{0.274,0.392,0.666}
\definecolor{red}{rgb}{0.627,0.117,0.156}
\definecolor{green}{rgb}{0,0.588,0.509}

\newcommand{\blue}[1]{{\color{blue}{#1\xspace}}}

\newcommand{\remove}[1]{}

\DeclareMathOperator{\pert}{pert}
\DeclareMathOperator{\skel}{skel}

\let\doendproof\endproof
\renewcommand{\endproof}{\hfill$\qed$\doendproof}

\newcommand{\minmaxface}{{\sc MinMaxFace}\xspace}
\newcommand{\kface}{$k$-{\sc MinMaxFace}\xspace}
\newcommand{\regdual}{{\sc UniformFaces}\xspace}

\newcommand{\pthreesat}{{\sc Planar $3$-Sat}\xspace}
\newcommand{\true}{\texttt{true}\xspace}
\newcommand{\false}{\texttt{false}\xspace}

\newcommand{\opt}{\textsc{Opt}}
\newcommand{\eps}{\ensuremath{\varepsilon}}

\newenvironment{sketch}{\noindent \textit{Sketch of Proof.}}{\hfill$\qed$\medskip}

\title{Planar Embeddings with Small and Uniform Faces\thanks{ Work by
    Giordano {Da Lozzo} was supported in part by the Italian Ministry
    of Education, University, and Research (MIUR) under PRIN
    2012C4E3KT national research project ``AMANDA -- Algorithmics for
    MAssive and Networked DAta''.  Work by Jan Kratochv\'il and V\'it
    Jel\'inek was supported by the grant no. 14-14179S of the Czech
    Science Foundation GA\v CR. Ignaz Rutter was supported by a
    fellowship within the Postdoc-Program of the German Academic
    Exchange Service (DAAD).}}

\author{Giordano Da Lozzo\inst{1}, V\'it Jel\'inek\inst{2}, Jan Kratochv\'il\inst{3},
  and Ignaz Rutter\inst{3,4}}

\institute{Department of Engineering, Roma Tre University, Italy\\
    \email{dalozzo@dia.uniroma3.it}
\and
    Computer Science Institute, Charles University, Prague\\ \email{jelinek@iuuk.mff.cuni.cz}
\and
    Department of Applied Mathematics, Charles University, Prague\\ \email{honza@kam.mff.cuni.cz}
\and
    Faculty of Informatics, Karlsruhe Institute of Technology (KIT), Germany, \\
    \email{rutter@kit.edu}}

\begin{document}
\pagestyle{plain}

\maketitle

\begin{abstract}
  Motivated by finding planar embeddings that lead to drawings with
  favorable aesthetics, we study the problems \minmaxface and \regdual
  of embedding a given biconnected multi-graph such that the largest
  face is as small as possible and such that all faces have the same
  size, respectively.

  We prove a complexity dichotomy for \minmaxface and show that
  deciding whether the maximum is at most $k$ is polynomial-time
  solvable for $k \le 4$ and NP-complete for $k \ge 5$.  Further, we
  give a 6-approximation for minimizing the maximum face in a planar
  embedding.  For \regdual, we show that the problem is NP-complete for odd $k \ge 7$
  and even $k \ge 10$. Moreover, we characterize the biconnected planar
  multi-graphs admitting 3- and 4-uniform embeddings (in a
  $k$-uniform embedding all faces have size~$k$) and give an efficient
  algorithm for testing the existence of a 6-uniform embedding.
\end{abstract}

\section{Introduction}

While there are infinitely many ways to embed a connected planar graph
into the plane without edge crossings, these embeddings can be grouped
into a finite number of equivalence classes, so-called
\emph{combinatorial embeddings}, where two embeddings are
\emph{equivalent} if the clockwise order around each vertex is the
same.  Many algorithms for drawing planar graphs require that the
input graph is provided together with a combinatorial embedding, which
the algorithm preserves.  Since the aesthetic properties of the
drawing
often depend critically on the chosen embedding, e.g. the number of
bends in orthogonal drawings, it is natural to ask for a planar
embedding that will lead to the best results.

In many cases the problem of optimizing some cost function over all
combinatorial embeddings is NP-complete.  For example, it is known
that it is NP-complete to test the existence of an embedding that
admits an orthogonal drawing without bends or an upward planar
embedding~\cite{gt-ccurpt-01}.  On the other hand, there are efficient
algorithms for minimizing various measures such as the radius of the
dual~\cite{adp-fmep-11,bm-ccvp-88} and attempts to minimize the number
of bends in orthogonal drawings subject to some
restrictions~\cite{bkrw-odfc-14,brw-oodc-13,dlv-sood-98}.

Usually, choosing a planar embedding is considered as deciding the
circular ordering of edges around vertices.  It can, however, also be
equivalently viewed as choosing the set of facial cycles, i.e., which
cycles become face boundaries.  In this sense it is natural to seek an
embedding whose facial cycles have favorable properties.  For example,
Gutwenger and Mutzel~\cite{gm-emmeea-03} give algorithms for computing
an embedding that maximizes the size of the outer face.  The most
general form of this problem is as follows.  The input consists of a
graph and a cost function on the cycles of the graph, and we seek a
planar embedding where the sum of the costs of the facial cycles is
minimum.  This general version of the problem has been introduced and
studied by Mutzel and Weiskircher~\cite{mw-ocep-99}.
Woeginger~\cite{w-epgmnlc-02} shows that it is NP-complete even when
assigning cost~0 to all cycles of size up to $k$ and cost~1 for longer
cycles.  Mutzel and Weiskircher~\cite{mw-ocep-99} propose an ILP
formulation for this problem based on SPQR-trees.

In this paper, we focus on two specific problems of this type, aimed
at reducing the visual complexity and eliminating certain artifacts
related to face sizes from drawings.  Namely, large faces in the
interior of a drawing may be perceived as holes and consequently
interpreted as an artifact of the graph.  Similarly, if the graph has
faces of vastly different sizes, this may leave the impression that
the drawn graph is highly irregular.  However, rather than being a
structural property of the graph, it is quite possible that the
artifacts in the drawing rather stem from a poor embedding choice and
can be avoided by choosing a more suitable planar embedding.

We thus propose two problems.  First, to avoid large faces in the
drawing, we seek to minimize the size of the largest face; we call
this problem \minmaxface.  Second, we study the problem of recognizing
those graphs that admit perfectly uniform face sizes; we call this
problem \regdual.  Both problems can be solved by the ILP approach of
Mutzel and Weiskircher~\cite{mw-ocep-99} but were not known to be
NP-hard.


\paragraph{Our Contributions.}

First, in Section~\ref{sec:minmaxface}, we study the computational
complexity of \minmaxface and its decision version $k$-\minmaxface,
which asks whether the input graph can be embedded such that the
maximum face size is at most~$k$.  We prove a complexity dichotomy for
this problem and show that $k$-\minmaxface is polynomial-time solvable
for $k \le 4$ and NP-complete for $k \ge 5$.  Our hardness result for
$k \ge 5$ strengthens Woeginger's result~\cite{w-epgmnlc-02}, which
states that it is NP-complete to minimize the number of faces of size
greater than $k$ for $k \ge 4$, whereas our reduction shows that it is
in fact NP-complete to decide whether such faces can be completely
avoided.  Furthermore, we give an efficient 6-approximation for
\minmaxface.

Second, in Section~\ref{sec:regular-duals-small}, we study the problem
of recognizing graphs that admit perfectly uniform face sizes
(\regdual), which is a special case of $k$-\minmaxface.  An embedding
is \emph{$k$-uniform} if all faces have size~$k$.  We characterize the
biconnected multi-graphs admitting a $k$-uniform embedding for $k=3,4$
and give an efficient recognition algorithm for $k=6$.  Finally, we
show that for odd $k \ge 7$ and even $k \ge 10$, it is NP-complete to
decide whether a planar graph admits a $k$-uniform embedding.

\section{Preliminaries}
\label{sec:preliminaries}

A graph $G = (V,E)$ is \emph{connected} if there is a path between any
two vertices.  A \emph{cutvertex} is a vertex whose removal
disconnects the graph.  A separating pair $\{u,v\}$ is a pair of
vertices whose removal disconnects the graph.  A connected graph is
\emph{biconnected} if it does not have a cutvertex and a biconnected
graph is \emph{3-connected} if it does not have a separating pair.
Unless specified otherwise, throughout the rest of the paper we will consider
graphs without loops, but with possible multiple edges.

We consider $st$-graphs with two special \emph{pole} vertices $s$ and
$t$.  The family of $st$-graphs can be constructed in a fashion very
similar to series-parallel graphs.  Namely, an edge $st$ is an
$st$-graph with poles $s$ and $t$.  Now let $G_i$ be an $st$-graph
with poles $s_i,t_i$ for $i=1,\dots,k$ and let $H$ be a planar graph
with two designated adjacent vertices $s$ and $t$ and $k+1$ edges $st,
e_1,\dots,e_k$.  We call $H$ the \emph{skeleton} of the composition
and its edges are called \emph{virtual edges}; the edge $st$ is the
\emph{parent edge} and $s$ and $t$ are the poles of the skeleton $H$.
To compose the $G_i$'s into an $st$-graph with poles $s$ and $t$, we
remove the edge $st$ from $H$ and replace each $e_i$ by $G_i$ for
$i=1,\dots,k$ by removing $e_i$ and identifying the poles of $G_i$
with the endpoints of $e_i$.  In fact, we only allow three types of
compositions: in a \emph{series composition} the skeleton $H$ is a
cycle of length~$k+1$, in a parallel composition $H$ consists of two
vertices connected by $k+1$ parallel edges, and in a \emph{rigid
  composition} $H$ is 3-connected.

For every biconnected planar graph $G$ with an edge $st$, the graph
$G-st$ is an $st$-graph with poles $s$ and $t$~\cite{dt-ogasp-90}.
Much in the same way as series-parallel graphs, the $st$-graph $G-st$
gives rise to a (de-)composition tree~$\mathcal T$ describing how it
can be obtained from single edges.  The nodes of $\mathcal T$
corresponding to edges, series, parallel, and rigid compositions of
the graph are \emph{Q-, S-, P-,} and \emph{R-nodes}, respectively.  To
obtain a composition tree for~$G$, we add an additional root Q-node
representing the edge $st$.  We associate with each node~$\mu$ the
skeleton of the composition and denote it by $\skel(\mu)$.  For a
Q-node~$\mu$, the skeleton consists of the two endpoints of the edge
represented by~$\mu$ and one real and one virtual edge between them
representing the rest of the graph.  For a node~$\mu$ of $\mathcal T$,
the \emph{pertinent graph} $\pert(\mu)$ is the subgraph represented by
the subtree with root~$\mu$.  For a virtual edge $\eps$ of a
skeleton~$\skel(\mu)$, the \emph{expansion graph} of $\eps$ is the
pertinent graph $\pert(\mu')$ of the neighbor $\mu'$ corresponding to
$\eps$ when considering $\mathcal T$ rooted at $\mu$.

The \emph{SPQR-tree} of $G$ with respect to the edge $st$, originally
introduced by Di Battista and Tamassia~\cite{dt-ogasp-90}, is the
(unique) smallest decomposition tree~$\mathcal T$ for $G$.  Using a
different edge $s't'$ of $G$ and a composition of $G-s't'$ corresponds
to rerooting $\mathcal T$ at the node representing $s't'$.  It thus
makes sense to say that $\mathcal T$ is the SPQR-tree of $G$.  The
SPQR-tree of $G$ has size linear in $G$ and can be computed in linear
time~\cite{gm-lis-01}.  Planar embeddings of $G$ correspond
bijectively to planar embeddings of all skeletons of $\mathcal T$; the
choices are the orderings of the parallel edges in P-nodes and the
embeddings of the R-node skeletons, which are unique up to a flip.
When considering rooted SPQR-trees, we assume that the embedding of
$G$ is such that the root edge is incident to the outer face, which is
equivalent to the parent edge being incident to the outer face in each
skeleton.
We remark that in a planar embedding of $G$, the poles of any node
$\mu$ of $\mathcal{T}$ are incident to the outer face of
$pert(\mu)$. Hence, in the following we only consider embeddings of
the pertinent graphs with their poles lying on the same face.

\section{Minimizing the Maximum Face}
\label{sec:minmaxface}

In this section we present our results on \minmaxface.  We first
strengthen the result of Woeginger~\cite{w-epgmnlc-02} and show that
$k$-\minmaxface is NP-complete for $k \ge 5$ and then present
efficient algorithms for $k=3,4$.  In particular, the hardness result
also implies that the problem \minmaxface is NP-hard.  Finally, we
give an efficient $6$-approximation for \minmaxface on biconnected
graphs. Recall that we allow graphs to have multiple edges. 



\begin{theorem}
  \label{thm:5-maxface-hardness}
  $k$-\minmaxface is NP-complete for any $k \ge 5$.
\end{theorem}

\begin{proof}
  Clearly, the problem is in NP, since we can simply guess a planar
  embedding and verify in polynomial time that all faces have size at
  most~$k$.

  We show hardness for $k=5$ and in the end briefly sketch how to
  adapt the proof for $k>5$.  We give a reduction from \pthreesat with
  the additional assumption that each variable occurs three times and
  each clause has size two or three.  Further, we can assume that if a
  variable occurs three times, then it appears twice as a positive
  literal and once as a negative literal.  This variant is
  NP-complete~\cite[Lemma 2.1]{fkmp-cimrp-95}.

  We construct gadgets where some of the edges are in fact two
  parallel paths, one consisting of a single edge and one of length~2
  or~3.  The ordering of these paths then decides which of the faces
  incident to the gadget edge is incident to a path of length~1 and
  which is incident to a path of length~2 or 3; see
  Fig.~\ref{fig:maxface-hardness}a.  Due to this use, we also call
  these gadgets $(1,2)$- and $(1,3)$-edges, respectively.

  \begin{figure}[tb]
    \centering
    \includegraphics{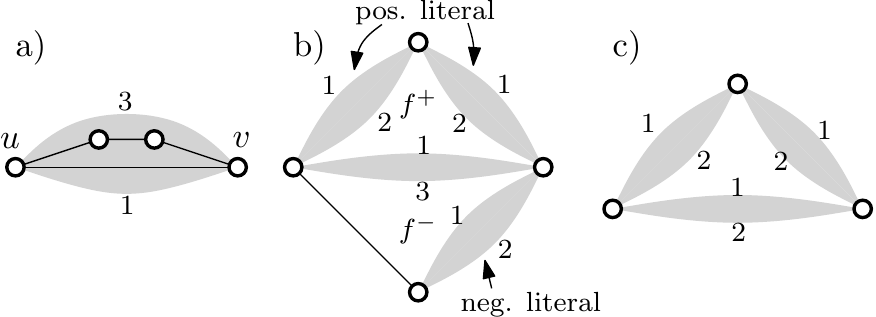}
    \caption{Illustration of the gadgets for the proof of
      Theorem~\ref{thm:5-maxface-hardness}.  (a) A $(1,3)$-edge. (b)~A
      variable gadget for a variable that occurs twice as a positive
      literal and once as a negative literal.  Changing the flip of
      the $(1,3)$-edge in the middle (variable edge) forces flipping
      the upper two literal edges.  (c) A clause gadget for a clause
      of size~3.}
    \label{fig:maxface-hardness}
  \end{figure}
  
  Now consider a variable $x$ whose positive literals occur $d^+$
  times.  Note that the negative literal hence occurs $3 - d^+$ times.
  We represent $x$ by a \emph{variable gadget} consisting of two
  cycles $C^+$ and $C^-$ of lengths~$5-d^+$ and $5 - (3 - d^+) = 2 +
  d^+$, respectively, sharing one edge.  The shared edge is actually a
  $(1,3)$-edge, called \emph{variable edge}, and in $C^+$ (in $C^-$),
  we replace $d^+$ of its edges ($3-d^+$ of its edges) by
  $(1,2)$-edges, called positive (negative) \emph{literal edges},
  respectively; see Fig.~\ref{fig:maxface-hardness}b.  We denote the
  faces bounded solely by $C^+$ and $C^-$ by $f^+$ and $f^-$,
  respectively.  Without loss of generality, we assume that the gadget
  is embedded so that $f^+$ and $f^-$ are inner faces, and we denote
  the outer face by $f_0$.  The gadget represents truth values as
  follows.  A literal edge represents the truth value \true if and
  only if its path of length~1 is incident to the outer face.  A
  variable edge represents value \true if and only if its path of
  length~1 is incident to $f^+$.  If the variable edge represents
  value true, then $f^-$ is incident to a path of length~3 of the
  variable edge.  Hence, all negative literal edges must transmit
  value \false.  A symmetric argument shows that if the variable edge
  encodes value \false, then all positive literal edges must transmit
  value \false.  On the other hand, given a truth value for variable
  $x$, choosing the flips of the variable edge and all literal edges
  accordingly yields an embedding where each inner face has size at
  most $5$.

  A clause gadget for a clause of size~3 consists of a cycle of three
  $(1,2)$-edges that correspond bijectively to the literals occurring
  in it; see Fig.~\ref{fig:maxface-hardness}c.  Similarly, a clause of
  size~2 consists of a cycle on four edges, two of which are
  $(1,2)$-edges corresponding to the literals.  The encoding is such
  that a literal edge has its path of length~2 incident to the inner
  face of the clause gadget if and only if such a literal has value
  \false.  Clearly, the inner face has size at most~5 if and only if
  at most two literals transmit value \false, otherwise the size is~6.
  Thus, the inner face of the clause gadget has size at most~5 if and
  only if at least one the literals transmit value~\true.

  We now construct a graph~$G_\varphi$ as follows.  We create for each
  variable a corresponding variable gadget and for each clause a
  corresponding clause gadget.  We then identify literal edges of
  variables and clauses that correspond to the same literal.  By
  adhering to the planar embedding of the variable--clause graph
  of~$\varphi$, the resulting graph $G_\varphi$ is planar and can be
  embedded such that all inner faces of the gadgets are faces of the
  graph.  Denote this plane graph by $H_\varphi$.  To obtain
  $G_\varphi$, we arbitrarily triangulate all faces of $H_\varphi$
  that are not internal faces of a gadget.  Then, the only embedding
  choices of $G_\varphi$ are the flips of the $(1,2)$- and
  $(1,3)$-edges.  We claim that $G_\varphi$ admits an embedding where
  every face has size at most~5 if and only if $\varphi$ is
  satisfiable.

  If $G$ is satisfiable, pick a satisfying truth assignment.  We flip
  each variable edge and each literal edge to encode its truth value
  in the assignment.  As argued above, every inner face of a variable
  now has size at most~5, and, since each clause contains at least one
  satisfied literal, also the inner faces of the clause gadgets have
  size at most~5.  Conversely, given a planar embedding of $G_\varphi$
  where each face has size at most~5, we construct a satisfying truth
  assignment for $\varphi$ by assigning a variable the truth value
  encoded by the variable edge in the corresponding gadget.  Due to
  the above properties, it follows that all edges corresponding to a
  negative literal must contribute a path of length~2 to each clause
  gadget containing such a literal.  However, each inner face of a
  clause gadget has only size~5, and hence at least one of the literal
  edges must contribute a path of only length~1, i.e., the clause
  contains a satisfied literal.  Since the construction of $G_\varphi$
  can clearly be done in polynomial time, this finishes the proof for
  $k=5$.

  For $k > 5$, it suffices to lengthen all cycles of the construction
  by $k-5$ edges.  All arguments naturally carry over.
\end{proof}

\subsection{Polynomial-Time Algorithm for Small Faces}
\label{sec:polyn-time-algor}

Next, we show that \kface is polynomial-time solvable for $k=3,4$.
Note that, if the input graph is simple, the problem for $k=3$ is
solvable if and only if the input graph is maximal planar.  A bit more
work is necessary if we allow parallel edges.

Let $G$ be a biconnected planar graph.  We devise a dynamic program on
the SPQR-tree $\mathcal T$ of $G$.  Let $\mathcal T$ be rooted at an
arbitrary Q-node and let $\mu$ be a node of $\mathcal T$.  
We call the clockwise and counterclockwise paths connecting the poles
of $\mu$ along the outer face the \emph{boundary paths} of
$\pert(\mu)$.  We say that an embedding of $\pert(\mu)$ has type
$(a,b)$ if and only if all its
inner faces have size at most~$k$ and its boundary paths have
length~$a$ and $b$, respectively.  Such an embedding is also called an
$(a,b)$-embedding.  We assume that $a \le b$.

Clearly, each of the two boundary paths of $\pert(\mu)$ in an embedding
$\mathcal E_\mu$ of type $(a,b)$ will be a proper subpath of the
boundary of a face in any embedding of $G$ where the embedding of
$\pert(\mu)$ is $\mathcal E_\mu$.  Hence, when seeking an embedding
where all faces have size at most~$k$, we are only interested in the
embedding $\mathcal E_\mu$ if $1 \le a \le b \le k-1$.  We define a
partial order on the embedding types by $(a',b') \preceq (a,b)$ if and
only if $a' \le a$ and $b' \le b$.  Replacing an
$(a,b)$-embedding~$\mathcal E_\mu$ of $\pert(\mu)$ by (a reflection
of) an $(a',b')$-embedding~$\mathcal E_\mu'$ with $(a',b') \preceq
(a,b)$ does not create faces of size more than $k$; all inner faces of
$\mathcal E_\mu'$ have size at most $k$ by assumption, and the only
other faces affected are the two faces incident to the two boundary
paths of $\mathcal E_\mu'$, whose length does not increase.  We thus
seek to compute for each node $\mu$ the minimal pairs $(a,b)$ for
which it admits an $(a,b)$-embedding.  We remark that $\pert(\mu)$ can
admit an embedding of type $(1,b)$ for any value of $b$ only if $\mu$
is either a P-node or a Q-node.  


\newcommand{\thmthreeminmaxtext}{3-\minmaxface can be solved in
  linear time for biconnected graphs.}


 We now present the algorithm for $k=3$, which works even if we allow
  parallel edges.

\begin{theorem}
  \label{thm:three-minmax}
  \thmthreeminmaxtext
\end{theorem}

\begin{proof}
  Clearly, the only interesting types of embeddings
  are $(1,1), (1,2)$ and $(2,2)$ and $\preceq$ defines a total
  ordering on them.  We thus seek to determine for each pertinent
  graph bottom-up in the SPQR-tree the smallest type (with respect to
  $\preceq$) of a valid planar embedding.  For Q-nodes this is
  $(1,1)$.  Now consider an R-node or S-node $\mu$.  By the above
  remark its only possible type of embedding can be $(2,2)$.  Since
  every face is bounded by at least three edges, it is not hard to see
  that $\pert(\mu)$ admits a $(2,2)$-embedding if and only if every
  face of $\skel(\mu)$ has size~3 and all children admit
  $(1,1)$-embeddings.

  For a P-node, we observe that none of its children can have a
  $(1,2)$-embedding, as no two P-node can be adjacent.  Thus, all
  children admit either a $(1,1)$-embedding, then they are Q-nodes, or
  they admit a $(2,2)$-embedding.  We denote the virtual edges in
  $\skel(\mu)$ by $(1,1)$-edges and $(2,2)$-edges, respectively,
  according to the type of embedding the corresponding graph admits.
  To obtain an embedding where all faces have size at most~3, we have
  to choose the embedding of $\skel(\mu)$ in such a way that every
  $(2,2)$-edge is adjacent to either two $(1,1)$-edges or to a
  $(1,1)$-edge and the parent edge.  Let $a$ and $b$ denote the number
  of $(1,1)$-edges and $(2,2)$-edges in $\skel(\mu)$, respectively.
  Clearly, an ordering satisfying these requirements exists if and
  only if $a \ge b-1$; otherwise we necessarily have two adjacent
  $(2,2)$-edges.  To find a good sequence, we proceed as follows.  If
  $a = b-1$, the sequence must alernatingly consist of $(2,2)$-edges
  and $(1,1)$-edges, starting with a $(1,1)$-edge.  The type of the
  resulting embedding is $(2,2)$ and one cannot do better.  If $a=b$,
  we do the same, but the type of the resulting embedding is $(1,2)$;
  again one cannot do better.  Finally, if $a \ge b+1$, we again do
  the same, and finally append the remaining $(1,1)$-edges.  Then the
  resulting embedding has type $(1,1)$.

  Clearly, we can process each node $\mu$ in time proportional to the
  size of its skeleton.  The graph admits an embedding if and only if
  the pertinent graph of the child of the root admits some valid
  embedding.
\end{proof}

We now deal with the case $k=4$, which is similar but more
complicated.  The relevant types are $(1,1), (1,2), (1,3), (2,2),
(2,3)$, and $(3,3)$.  We note that precisely the two elements $(2,2)$
and $(1,3)$ are incomparable with respect to $\preceq$.  Thus, it
seems that, rather than computing only the single smallest type for
which each pertinent graph admits an embedding, we are now forced to
find all minimum pairs for which the pertinent graph admits a
corresponding embedding.  However, by the above observation,
if a pertinent graph $\pert(\mu)$ admits a $(1,3)$-embedding, then $\mu$ must
be a P-node.  However, if the parent of $\mu$ is an S-node or an
R-node, then using a $(1,3)$-embedding results in a face of size at
least~5.  Thus, such an embedding can only be used if the parent is
the root Q-node.  If there is the choice of a $(2,2)$-embedding in
this case, it can of course also be used at the root.  Therefore, we
can mostly ignore the $(1,3)$-case and consider the linearly ordered
embedding types $(1,1), (1,2), (2,2), (2,3)$ and $(3,3)$.  The type
$(1,3)$ is only relevant for P-nodes whose pertinent graph admits an
embedding of type $(1,3)$ embedding but no embedding of type $(2,2)$.

\remove{The running
time stems from the fact that, for an R-node, we need to find a
matching between the virtual edges whose expansion graphs admit a
$(1,2)$-embedding and the incident triangular faces of the skeleton.
}



\begin{theorem}
  \label{thm:4-minmaxface}
  4-\minmaxface can be solved in $O(n^{1.5})$ time for biconnected
  graphs.
\end{theorem}

\begin{proof}
  We process the SPQR-tree of the input graph in a bottom-up fashion.
  The pertinent graphs of Q-nodes admit embeddings of type $(1,1)$.

  Now consider an S- or an R-node $\mu$.  All faces of $\skel(\mu)$
  must have size at most~4.  Moreover, since all faces have length at
  least~3, a valid embedding of $\pert(\mu)$ does not exist if some
  child only allows embeddings of type $(1,3)$, $(2,3)$ or $(3,3)$.
  Thus, the only freedom is to choose the flips of the pertinent
  graphs admitting only $(1,2)$-embeddings.  A face can receive only a
  single path of length~2 from one of its incident edges, and this is
  possible only if the face is a triangle and none of its incident
  edges is a $(2,2)$-edge.  We thus seek a matching between the
  $(1,2)$-edges and their incident faces that can receive a path of
  length~2.  Depending on the size of the faces incident to the parent
  edge and whether they need to receive a path of length~2 in order to
  find a valid embedding, the type is either $(2,2)$ (if both faces
  are triangles and they do not need to receive a path of length~2),
  $(2,3)$ (if one is a triangle that does not need to receive a path
  of length~2) or $(3,3)$ (remaining cases).

  Now consider a P-node.  Each child must have an embedding of type
  $(1,1)$, $(2,2)$, $(2,3)$ or $(3,3)$.  Again, we denote the edges
  whose corresponding pertinent graph admits an embedding of type
  $(a,b)$ as $(a,b)$-edges.

  First observe that removing in any embedding all $(2,2)$-edges
  except for one and placing them next to the single $(2,2)$-edge we
  did not remove results in a valid embedding whose boundary paths do
  not increase.  Thus, we can assume without loss of generality that
  there is at most one $(2,2)$-edge.  Moreover, if there is a
  $(2,3)$-edge, we can actually move the $(2,2)$-edge next to it
  without increasing the size of any face.  Thus, if there are any
  $(2,3)$-edges we can even assume that there is no $(2,2)$-edge.

  Let us first assume that there is no $(2,3)$-edge.  We then have to
  choose the embedding such that $(1,1)$-edges alternate with
  $(3,3)$-edges and the single $(2,2)$-edge.  We append any excess of
  $(1,1)$-edges at the end.  Let $a$ denote the number of
  $(1,1)$-edges and let $b$ denote the total number of $(2,2)$ and
  $(3,3)$-edges.  A valid embedding exists only if $a \ge b-1$.  In
  this case a suitable sequence always exists.  If possible, we start
  and end with a $(1,1)$-edge, resulting in a $(1,1)$-embedding.  If
  this is not the case, we try to start with a $(1,1)$ and put the
  $(2,2)$ in the end if it exists.  Then we obtain a $(1,2)$-embedding
  if there is a $(2,2)$-edge and a $(1,3)$-embedding otherwise.  If
  this is also not possible since $a = b-1$, we start with the
  $(2,2)$-edge if it exists.  This results in either a $(2,3)$ or a
  $(3,3)$-embedding.

  The bottleneck concerning the running time is finding the matching
  for treating the R-node, which can be solved in $O(n^{1.5})$
  time~\cite{g-ertdc-83}.
\end{proof}

\subsection{Approximation Algorithm}
\label{sec:approx-algo}

In this section, we present a constant-factor approximation algorithm for the
problem of minimizing the largest face in an embedding of a biconnected graph
$G=(V,E)$. We again solve the problem by dynamic programming on the SPQR-tree
of~$G$.

Let $G$ be a biconnected planar graph, and let $\mathcal T$ be its SPQR-tree,
rooted at an arbitrary Q-node. Let $\mu$ be a node of $\mathcal T$. We shall
consider the embeddings of $\pert(\mu)$ where the two poles are embedded on the
outer face. We also include the parent edge in the embedding, by drawing it in
the outer face. In such an embedding of $\skel(\mu)$, the two faces incident to
the parent edge are called \emph{the outer faces}, while the remaining faces
are \emph{inner faces}. 

Recall that an \emph{$(a,b)$-embedding} of $\pert(\mu)$ is an embedding whose
boundary paths have lengths $a$ and $b$, where we always assume that $a\le b$.
We say that an $(a,b)$-embedding of $\pert(\mu)$ is \emph{out-minimal} if for
any $(a',b')$-embedding of $\pert(\mu)$, we have $a\le a'$ and $b\le b'$. Note
that an out-minimal embedding need not exist; e.g., $\pert(\mu)$ may admit a
$(2,4)$-embedding and a $(3,3)$-embedding, but no $(a,b)$-embedding with $a\le
2$ and $b\le 3$. We will later show, however, that such a situation can only
occur when $\mu$ is an S-node.

Let $\opt(G)$ denote the smallest integer $k$ such that $G$ has an embedding
whose every face has size at most~$k$. For a node $\mu$ of $\mathcal T$, we say
that an embedding of $\pert(\mu)$ is \emph{$c$-approximate}, if each inner face
of the embedding has size at most $c\cdot\opt(G)$. 

Call an embedding of $\pert(\mu)$ \emph{neat} if it is out-minimal and
6-approximate. The main result of this section is the next proposition.

  \begin{proposition} \label{pro:approx}Let $G$ be a biconnected planar graph with SPQR
    tree $\mathcal T$, rooted at an arbitrary $Q$-node. Then the
    pertinent graph of every Q-node, P-node or R-node of $\mathcal T$
    has a neat embedding, and this embedding may be computed in
    polynomial time.
  \end{proposition}
Since the pertinent graph of the root of $\mathcal T$ is the whole graph $G$,
the proposition implies a polynomial 6-approximation algorithm for minimization
of largest face.

Our proof of Proposition~\ref{pro:approx} is constructive. Fix a node $\mu$
of $\mathcal T$ which is not an S-node.  We now describe an algorithm that
computes a neat embedding of $\pert(\mu)$, assuming that neat embeddings are
available for the pertinent graphs of all the descendant nodes of $\mu$ that are
not $S$-nodes. We distinguish cases based on the type of the
node~$\mu$.

\paragraph{Non-root Q-nodes.}
As a base case, suppose that $\mu$ is a non-root Q-node of $\mathcal T$. Then
$\pert(\mu)$ is a single edge, and its unique embedding is clearly neat. 

\paragraph{P-nodes.}
Next, suppose that $\mu$ is a P-node with $k$ child nodes $\mu_1,\dotsc,\mu_k$,
represented by $k$ skeleton edges $e_1,\dotsc,e_k$. Let $G_i$ be the expansion
graph of~$e_i$. We construct the \emph{expanded skeleton} $\skel^*(\mu)$ as
follows: if for some $i$ the child node $\mu_i$ is an S-node whose skeleton is a
path of length $m$, replace the edge $e_i$ by a path of length $m$, whose
edges correspond in a natural way to the edges of $\skel(\mu_i)$. 

Every edge $e'$ of the expanded skeleton corresponds to a node $\mu'$ of
$\mathcal T$ which is a child or a grand-child of~$\mu$. Moreover, $\mu'$ 
is not an S-node, and we may thus assume that we have already computed a neat
embedding for~$\pert(\mu')$. Note that $\pert(\mu')$ is the expansion graph
of~$e'$.

For each $i\in\{1,\dotsc,k\}$ define $\ell_i$ to be the smallest value
such that $G_i$ has an embedding with boundary path of length~$\ell_i$. We
compute $\ell_i$ as follows: if $\mu_i$ is not and S-node, then we already
know a neat $(a_i,b_i)$-embedding of $G_i$, and we may put $\ell_i=a_i$. If, on
the other hand, $\mu_i$ is an S-node, then let $m$ be the number of edges in the
path $\skel(\mu_i)$, and let $G_i^1, G_i^2,\dotsc, G_i^m$ be the expansion
graphs of the edges of the path. For each $G_i^j$, we have already computed a
neat $(a_j,b_j)$-embedding, so we may now put $\ell_i=\sum_{j=1}^m a_j$.
Clearly, this value of $\ell_i$ corresponds to the definition given above.

We now fix two distinct indices $\alpha, \beta\in\{1,\dotsc,k\}$, so that
the values $\ell_\alpha$ and $\ell_\beta$ are as small as possible; formally,
$\ell_\alpha=\min\{\ell_i;\; i=1,\dotsc,k\}$ and $\ell_\beta=\min\{\ell_i;\;
i=1,\dotsc,k \text{ and } i\neq\alpha\}$. 

Let us fix an embedding of $\skel(\mu)$ in which $e_\alpha$ and $e_\beta$ are
adjacent to the outer faces. We extend this embedding of $\skel(\mu)$ into an
embedding of $\pert(\mu)$ by replacing each edge of $\skel^*(\mu)$ by a neat
embedding of its expansion graph, in such a way that the two boundary paths
have lengths $\ell_\alpha$ and~$\ell_\beta$. Let $\mathcal E$ be the resulting 
 $(\ell_\alpha,\ell_\beta)$-embedding of $\pert(\mu)$.

We now show that $\mathcal E$ is neat. From the definitions of $\ell_\alpha$ and
$\ell_\beta$, we easily see that $\mathcal E$ is out-minimal. It remains to
show that it is 6-approximate. Let $f$ be any inner face of $\mathcal E$. If
$f$ is an inner face of the expansion graph $G_i$ of some $e_i$, then $f$ is an
inner face of some previously constructed neat embedding, hence $|f|\le
6\cdot\opt(G)$. 

Suppose then that $f$ is not the inner face of any~$G_i$. Then the boundary of
$f$ intersects two distinct expansion graphs $G_i$ and $G_j$. Hence the boundary
of $f$ is the union of two paths $P_i$ and $P_j$, with $P_i\subseteq G_i$ and
$P_j\subseteq G_j$. Let $d_i$ and $d_j$ be the lengths of $P_i$ and $P_j$,
respectively, and assume that $d_i\le d_j$. It follows that $|f|=d_i+d_j\le
2d_j$. We claim that every embedding of $G$ has a face of size at least
$d_j/2$. If $\mu_j$ is not an S-node, this follows from the fact that $P_j$ is a
boundary path in an out-minimal embedding of $G_j$, hence any other embedding of
$G_j$ must have a boundary path of length at least $d_j$. If, on the other hand,
$\mu_j$ is an S-node, then in every embedding of $G_j$, the two boundary paths
have total length at least $d_j$, so every embedding of $G_j$ has a boundary
path of length at least $d_j/2$ and thus $G$ has a face of size at
least~$d_j/2$. We conclude that $|f|\le 2d_j\le 4\cdot\opt(G)$, showing that
$\mathcal E$ is indeed neat.

\paragraph{R-nodes.} Suppose now that $\mu$ is an R-node. As with P-nodes, we
define the \emph{expanded skeleton} $\skel*(\mu)$ by replacing each edge of
$\skel(\mu)$ corresponding to an S-node by a path of appropriate length. The
graph $\skel*(\mu)$ together with the parent edge forms a subdivision of a
3-connected graph. In particular, its embedding is determined uniquely up to a
flip and a choice of outer face. Fix an embedding of $\skel^*(\mu)$ and the
parent edge, so that the parent edge is on the outer face. Let $f_1$ and $f_2$
be the two faces incident to the parent edge of~$\mu$. 

Let $e$ be an edge of $\skel^*(\mu)$, let $G_e$ be its expansion graph, and let
$\mathcal E_e$ be a neat $(a,b)$-embedding of $G_e$, for some $a\le b$. The
boundary path of $\mathcal E_e$ of length $a$ will be called \emph{the short
side} of $\mathcal E_e$, while the boundary path of length $b$ will be \emph{the
long side}. If $a=b$, we choose the long side and short side arbitrarily.

Our goal is to extend the embedding of $\skel^*(\mu)$ into an embedding of
$\pert(\mu)$ by replacing each edge $e$ of $\skel^*(\mu)$ with a copy of
$\mathcal E_e$. In doing so, we have to choose which of the two faces incident
to $e$ will be adjacent to the short side of $\mathcal E_e$. 

First of all, if $e$ is an edge of $\skel^*(\mu)$ incident to one of the outer
faces $f_1$ or $f_2$, we embed $\mathcal E_e$ in such a way that its short side
is adjacent to the outer face. Since $f_1$ and $f_2$ do not share an edge in
$\skel^*(\mu)$, such an embedding is always possible, and guarantees that the
resulting embedding of $\pert(\mu)$ will be out-minimal.  

It remains to determine the orientation of $\mathcal E_e$ for the edges $e$
that are not incident to the outer faces, in such a way that the largest face
of the resulting embedding will be as small as possible. Rather than solving
this task optimally, we formulate a linear programming relaxation, and then
apply a rounding step which will guarantee a constant factor approximation.

Intuitively, the linear program works as follows: given an edge $e$ incident to
a pair of faces $f$ and $g$, and a corresponding graph $G_e$ with a short side
of length $a$ and a long side of length $b$, rather than assigning the short
side to one face and the long side to the other, we assign to each of the two
faces a fractional value in the interval $[a,b]$, so that the two values
assigned by $e$ to $f$ and $g$ have sum $a+b$, and the maximum total amount
assigned to a single face of $\skel^*(\mu)$ from its incident edges is as small
as possible.

More precisely, we consider the linear program with the set of variables
\[
\{M\}\cup\{x_{e,f};\; e\text{ is an edge adjacent to face }f\},
\]
where the goal is to minimize $M$ subject to the following constraints:
\begin{itemize}
 \item For every edge $e$ adjacent to a pair of faces $f$ and $g$, we have the
constraints $x_{e,f}+x_{e,g}=a+b$, $a\le x_{e,f}\le b$ and $a\le x_{e,g}\le b$,
where $a\le b$ are the lengths of the two boundary paths of $\mathcal E_e$.
\item Moreover, if an edge $e$ is adjacent to an outer face $f\in\{f_1,f_2\}$
as well as an inner face $g$, then we set $x_{e,f}=a$ and $x_{e,g}=b$, with $a$
and $b$ as above.
\item For every inner face $f$ of $\skel^*(\mu)$, we have the constraint $\sum_e
x_{e,f}\le M$, where the sum is over all edges incident to~$f$.
\end{itemize}

Given an optimal solution of the above linear program, we determine the
embedding of $\pert(\mu)$ as follows: for an edge $e$ of $\skel^*(\mu)$ incident
to two inner faces $f$ and $g$, if $x_{e,f}\le x_{e,g}$, embed $\mathcal E_e$
with its short side incident to $f$ and long side incident to $g$. Let
$\mathcal{E}_\mu$ be the resulting embedding.

We claim that $\mathcal E_\mu$ is neat. We have already seen that $\mathcal
E_\mu$ is out-minimal, so it remains to show that every inner face of
$\mathcal{E}_\mu$ has size at most $6\cdot\opt(G)$. Let us say that an inner
face of $\mathcal E_\mu$ is \emph{deep} if it is also an inner face of some
$\mathcal E_e$, and it is \emph{shallow} if it corresponds to a face of
$\skel^*(\mu)$. Note that the deep faces have size at most $6\cdot\opt(G)$,
since all the $\mathcal{E}_e$ are neat embeddings, so we only need to estimate
the size of the shallow faces.

Let $\opt_{out}$ denote the minimum $k$ such that $\pert(\mu)$ has an
out-minimal embedding whose every shallow face has size at most~$k$. We claim
that $\opt_{out}\le 3\cdot\opt(G)$. To see this, consider an embedding of
$\pert(\mu)$ in which each face has size at most $\opt(G)$. In this embedding,
replace each subembedding of $G_e$ by a copy of $\mathcal E_e$, without
increasing the size of any shallow face. This can be done, because each
$\mathcal E_e$ is out-minimal. Call the resulting embedding $\mathcal E'$. Next,
for every edge $e$ of $\skel^*(\mu)$ adjacent to $f_1$ or $f_2$, flip
$\mathcal E_e$ so that its short side is incident to $f_1$ or $f_2$. Let
$\mathcal E''$ be the resulting embedding of $\pert(\mu)$. Clearly, $\mathcal
E''$ is out-minimal.

In $\mathcal E''$, some inner shallow face $f$ adjacent to $f_1$ or $f_2$ may
have larger size than the corresponding face of $\mathcal E'$; however, for such
an $f$, its size in $\mathcal E''$ is at most equal to the sum of the sizes of
$f$, $f_1$ and $f_2$ in $\mathcal E'$. In particular, each inner shallow face
has size at most $3\cdot \opt(G)$ in $\mathcal E''$, and hence $\opt_{out}\le
3\cdot\opt(G)$, as claimed. 

We will now show that each shallow face of $\mathcal E_\mu$ has size at most
$2\cdot\opt_{out}$. Let $M$ be the value of optimum solution to the linear
program defined above. Clearly, $M\le \opt_{out}$, since from an out-minimal
embedding with shallow faces of size at most $\opt_{out}$, we may directly
construct a feasible solution of the linear program with value $\opt_{out}$. 
Let $f$ be a shallow face of $\mathcal E_\mu$. Let $e$ be an edge of
$\skel^*(\mu)$ incident to $f$, and let $g$ be the other face incident to~$e$. 
Let $a$ and $b$ be the lengths of the short side and long side of $\mathcal
E_e$, respectively. If $x_{e,f}\le x_{e,g}$, then $\mathcal E_e$ contributes to
the boundary of $f$ by its short side, which has length $a$. Otherwise, $f$ has
the long side of $\mathcal E_e$ on its boundary, but that may only happen when
$x_{e,f}\ge x_{e,g}$, and hence $b\le a+b=x_{e,f}+x_{e,g}\le 2x_{e,f}$. From
this, we see that $f$ has size at most $\sum_e 2x_{e,f}\le 2M$, with the
previous sum ranging over all edges of $\skel^*(\mu)$ incident to~$f$. 

Thus, for every shallow face $f$ of $\mathcal E_\mu$, we have $|f|\le 2M\le
2\cdot\opt_{out}\le 6\cdot \opt(G)$, showing that $\mathcal E_\mu$ is neat.

\paragraph{The root Q-node.} Finally, suppose that $\mu$ is the root of the
SPQR-tree $\mathcal T$. That means that $\mu$ is a Q-node, and its skeleton is
formed by two parallel edges $e_1$ and $e_2$, where the expansion graph of $e_1$
is a single edge and the expansion graph $G_2$ of $e_2$ is the pertinent graph
of the unique child node $\mu'$ of~$\mu$. If $\mu'$ is not an S-node, we already
have a neat $(a,b)$-embedding $\mathcal E_2$ of $G_2$, and by inserting the
edge $e_1$ to this embedding in such a way that the outer face has size $a+1$,
we clearly obtain a neat embedding of~$G$. If $\mu'$ is an S-node, then $G_2$
is a chain of biconnected graphs $G_2^1, G_2^2,\dotsc, G_2^k$, and for each
$G_2^i$ we have a neat $(a_i,b_i)$-embedding. Combining these embedding in an
obvious way, and adding the edge $e_1$, we get an embedding of $G$ whose outer
face has size $1+a_1+a_2+\dotsb+a_k$, and whose unique inner shallow face has
size $1+b_1+b_2+\dotsb+b_k$. Since in each embedding of $G$, the two faces
incident to $e_1$ have total size at least $2+a_1+\dotsb+a_k+b_1+\dotsb+b_k$,
we conclude that our embedding of $G$ is neat.

This completes the proof of Proposition~\ref{pro:approx}, and yields a
6-approximation algorithm for the minimization of largest face in biconnected
graphs.

\begin{theorem}
  A $6$-approximation for \minmaxface in biconnected graphs can be
  computed in polynomial time.
\end{theorem}

\section{Perfectly Uniform Face Sizes}
\label{sec:regular-duals-small}

In this section we study the problem of deciding whether a biconnected
planar graph admits a $k$-uniform embedding.  Note that, due to
Euler's formula, a connected planar graph with $n$ vertices and $m$
edges has $f = m-n+2$ faces.  In order to admit an embedding where
every face has size~$k$, it is necessary that $2m = fk$.  Hence there
is at most one value of $k$ for which the graph may admit a
$k$-uniform embedding.

In the following, we characterize the graphs admitting
3-uniform and 4-uniform embeddings, and we give an efficient algorithm
for testing whether a graph admits a 6-uniform embedding.  Finally, we
show that testing whether a graph admits a $k$-uniform embedding is
NP-complete for odd $k \ge 7$ and even $k \ge 10$.  We leave open the
cases $k=5$ and $k=8$.

Our characterizations and our testing algorithm use the recursive
structure of the SPQR-tree.  To this end, it is necessary to consider
embeddings of pertinent graphs, where we only require that the
interior faces have size~$k$, whereas the outer face may have
different size, although it must not be too large.  We call such an
embedding \emph{almost $k$-uniform}.  The following lemma states that
the size of the outer face in such an embedding depends only on the
number of vertices and edges in the pertinent graph.

\newcommand{\lemouterfacelengthtext}{Let $G$ be a graph with $n$
  vertices and $m$ edges with an almost $k$-uniform embedding.  Then
  the outer face has length~$\ell = k(n-m-1) + 2m$.}


\begin{lemma}
  \label{lem:outer-face-length}
  \lemouterfacelengthtext
\end{lemma}

\begin{proof}
  Let $f$ denote the number of faces of $G$ in a planar embedding,
  which us uniquely determined by Euler's formula $n-m+f=2$.  By
  double counting, we find that $(f-1) \cdot k + \ell = 2m$.  Euler's
  formula implies that $f=2+m-n$, and plugging this into the second
  formula, we obtain that $(1+m-n) \cdot k + \ell = 2m$ or,
  equivalently, $\ell = k(n-m-1) + 2m$.
\end{proof}

Thus, for small values of $k$, where the two boundary paths of the
pertinent graph may have only few different lengths, the type of an
almost $k$-uniform embedding is essentially fixed.

\subsection{Characterization for $k=3,4$}


For 3-uniform embeddings first observe that every facial cycle must be
a triangle.  If the input graph is simple, then this implies that it
must be a triangulation.  Then the graph is 3-connected and the planar
embedding is uniquely determined.  We characterize the multi-graphs
that have such an embedding.

\remove{
Using this fact, the biconnected graphs (with possible multiple edges) admitting a
$k$-uniform dual for $k=3,4$ can be characterized. \remove{; we omit the details.  }
We remark that the simple graphs admitting 3-uniform and 4-uniform embeddings
are precisely the maximal planar graphs and the maximal planar bipartite
graphs.}

  \begin{theorem}
    \label{thm:3-uniform}
    A biconnected planar graph $G$ admits $3$-uniform embedding if and
    only if its SPQR-tree satisfies all of the following conditions.

    \begin{compactenum}[(i)]
    \item S- and R-nodes are only adjacent to Q- and P-nodes.
    \item Every R-node skeleton is a planar triangulation.
    \item Every S-node skeleton has size~3.
    \item Every P-node with $k$ neighbors has $k$ even and precisely
      $k/2$ of the neighbors are Q-nodes.
    \end{compactenum}
  \end{theorem}

\begin{proof}
  It is not hard to see that all conditions are necessary.  We prove
  sufficiency.  To this end, we choose the embeddings of the R-node
  skeletons arbitrarily, and we embed the P-node skeletons such that
  virtual edges corresponding to Q-node and non-Q-node neighbors
  alternate.  We claim that in the resulting planar embedding of $G$
  all faces have size~3.

  To this end, root the SPQR-tree $\mathcal T$ of $G$ at an arbitrary
  edge $e$ and consider the embedding $e$ incident to the outer face.

  \begin{claim}
    \begin{compactenum}
    \item If $\mu$ is a Q-node or a P-node whose parent is not a
      Q-node, then it has an almost 3-uniform embedding of type
      $(1,1)$.
    \item If $\mu$ is an S-node, an R-node, or a P-node whose parent
      is a Q-node, then it has an almost 3-uniform embedding of type
      $(2,2)$.
    \end{compactenum}
  \end{claim}
  
  We prove this by induction on the height of the node in the
  SPQR-tree.  Clearly, the statement holds for Q-nodes.  Now consider
  an internal node~$\mu$ and assume that the claim holds for all
  children.

  If $\mu$ is an S-node, then the clockwise (counterclockwise) path of
  $\pert(\mu)$ between the poles along the outer face is the
  concatenation of the clockwise path (counterclockwise) paths of the
  pertinent graphs of its children.  By property~(iii) there are only
  two children and by property~(i) they are either Q- or P-nodes.  By
  the inductive hypothesis, their embeddings are almost 3-uniform and
  have type $(1,1)$, and hence the type of the embedding of
  $\pert(\mu)$ is $(2,2)$.

  If $\mu$ is an R-node, then its clockwise (counterclockwise) path
  between the poles is the concatenation of the clockwise
  (counterclockwise) paths of the pertinent graphs corresponding to
  the edges on the clockwise (counterclockwise) path between the
  poles.  By property~(ii) each of these paths has length~2 in
  $\skel(\mu)$ and the children are either Q- or P-nodes.  Thus, by
  the inductive hypothesis, their embeddings have type $(1,1)$.

  If $\mu$ is a P-node whose parent is not a Q-node, then, by our
  choice of the planar embedding, the two outer paths in $\skel(\mu)$
  are edges corresponding to Q-nodes, and the claim follows from the
  inductive hypothesis.  If the parent of $\mu$ is a Q-node, then,
  again by the embedding choice, the two edges outer paths in
  $\skel(\mu)$ are edges corresponding to S- or R-nodes, and again the
  inductive hypothesis implies the claim.  This finishes the proof of
  the claim.

  Let now $\mu$ denote the Q-node corresponding to the root edge $e$
  and consider the two faces incident to $e$, which show up as faces
  in $\skel(\mu)$.  Let $\mu'$ be the neighbor of $\mu$ in the
  SPQR-tree.  Then $\mu'$ is either an S-node, an R-node, or a P-node
  whose parent is a Q-node.  In all cases the embedding of
  $\pert(\mu')$ has type $(2,2)$, and hence the two faces incident to
  $e$ have size~3.  Since $e$ was chosen arbitrarily, it follows that
  each face has size~3.
\end{proof}

\begin{corollary}
  It can be tested in linear time whether a biconnected planar graph
  admits a 3-regular dual.
\end{corollary}

For 4-uniform embeddings observe that every facial cycle must be a
simple cycle of length~4.  Since every planar graph containing a cycle
of odd length also has a face of odd length in any planar embedding,
it follows that the graph must be bipartite.

Now, if the graph is simple, the graph must be planar, bipartite and
each face must have size~4.  It is well known (and follows from
Euler's formula) that this is the case if and only if the graph has
$2n-4$ edges; the maximum number of edges for a simple bipartite
planar graph.  Again, if the graph is not simple more work is
necessary.  For a virtual edge $e$ in a skeleton $\skel(\mu)$, we
denote by $m_e$ and $n_e$ the number of edges in its expansion graph.

  \begin{theorem}
    \label{thm:4-uniform}
    A biconnected planar graph admits a 4-regular dual if and only if
    it is bipartite and satisfies the following conditions.

    \begin{compactenum}[(i)]
    \item For each P-node either all expansion graphs satisfy $m_e =
      2n_e-4$, or half of them satisfy $m_e=2n_e-5$ and the other half
      are Q-nodes.
    \item For each S- or R-node all faces have size~3 or 4; the
      expansion graphs of all edges incident to faces of size~4
      satisfy $m_e = 2n_e-3$ and for each triangular face, there is
      precisely one edge whose expansion graph satisfies $m_e = 2n_e -
      4$, the others satisfy $m_e = 2n_e-3$.
    \end{compactenum}
  \end{theorem}

\begin{proof}
  We choose the planar embedding as follows.  For each P-node, if half
  of the neighbors are Q-nodes, then we choose the embedding such that
  Q-nodes and non-Q-nodes alternate.  All remaining embedding choices
  can be done arbitrarily.  We claim that in the resulting embedding
  all faces have size~4.

  As in the proof of Theorem~\ref{thm:3-uniform}, root the SPQR-tree
  $\mathcal T$ of $G$ at an arbitrary edge $e$ and consider the
  embedding as having $e$ incident to the outer face.

  \begin{claim}
    For each node $\mu$ of $\mathcal T$ in the embedding of
    $\pert(\mu)$ without the parent edge denote by $\ell_\mu$ and
    $r_\mu$ the length of the clockwise and counterclockwise path on
    the outer face connecting the poles of $\mu$.

    \begin{compactenum}
    \item Each internal face of $\pert(\mu)$ has size 4.
    \item If $\mu$ is a Q-node or a P-node with Q-node neighbors whose
      parent is not a Q-node, then $\pert(\mu)$ has an almost
      4-uniform embedding of type $(1,1)$.

    \item If $\mu$ is a P-node whose neighbors all satisfy $m_e = 2n_e
      -4$, or $\mu$ is an S- or an R-node whose parent satisfies $m_e
      = 2n_e - 4$, then $\pert(\mu)$ has an almost 4-uniform embedding
      of type $(2,2)$.

    \item If $\mu$ is a P-node with Q-node neighbors whose parent is a
      Q-node, or if $\mu$ is an S- or an R-node whose parent is a
      Q-node or satisfies $m_e = 2n_e - 3$, then $\pert(\mu)$ has an
      almost 4-uniform embedding of type $(3,3)$.
    \end{compactenum}
  \end{claim}

  The proof of the claim is by structural induction on the SPQR-tree.
  Clearly, it holds for the leaves, which are Q-nodes.  Now consider
  an internal node $\mu$.

  If $\mu$ is a P-node, with Q-node neighbors whose parent is not a
  Q-node then, by property (i) all children have almost 4-uniform
  embeddings.  Further, the children that are not Q-nodes satisfy $m_e
  = 2n_e - 4$, and hence, their outer face has size~6 by
  Lemma~\ref{lem:outer-face-length}.  It must hence be an S- or an
  R-node and, by the inductive hypothesis, their embeddings have type
  $(3,3)$.  Thus, the alternation of Q-nodes and these children
  ensures that inner faces have size~4.  Moreover, since the parent is
  not a Q-node, the linear ordering of the children (excluding the
  parent) starts and ends with a Q-node.  Hence the embedding of
  $\pert(\mu)$ has type $(1,1)$.
  
  If $\mu$ is a P-node whose neighbors all satisfy $m_e = 2n_e - 4$,
  then all children have almost 4-uniform embeddings whose outer faces
  have size~4 by Lemma~\ref{lem:outer-face-length}.  Since a
  non-P-node cannot have an embedding of type $(1,x)$ for any value of
  $x$, their embeddings have type $(2,2)$.  This implies the inductive
  hypothesis.

  If $\mu$ is a P-node with Q-node neighbors whose parent is a Q-node,
  then one more than half of its children satisfy $m_e = 2n_e-4$ and
  hence have almost 4-uniform embeddings of type $(3,3)$.  The
  alternation of Q-nodes and these children implies the statement.

  If $\mu$ is an S- or an R-node whose parent satisfies $m_e = 2n_e
  -3$ (or it is a Q-node), the internal faces have size~4 according to
  the inductive hypothesis and property~(ii).  A similar argument
  shows that the embedding has type $(3,3)$.

  If $\mu$ is an S- or an R-node whose parent satisfies $m_e = 2n_e -
  4$, then the two faces incident to the parent edge are triangles,
  and the children all satisfy $m_e = 2n_e - 3$, and hence have almost
  4-uniform embeddings of type $(1,1)$.  Thus the embedding of
  $\pert(\mu)$ has type $(2,2)$.  This finishes the proof of the
  claim, and as it immediately implies that every face has size~4,
  also the proof of the theorem.
\end{proof}

\begin{corollary}
  It can be tested in linear time whether a biconnected planar graph
  admits a 4-regular dual.
\end{corollary}


\subsection{Testing Algorithm for 6-Uniform Embeddings}
\label{sec:6-regular-dual}

  To test the existence of a 6-uniform embedding, we again use
  bottom-up traversal of the SPQR-tree and are therefore interested in
  the types of almost 6-uniform embeddings of pertinent graphs.
  Clearly, each of the two boundary paths of a pertinent graph, may
  have length at most~5.  Thus, only embedding of type $(a,b)$ with $1
  \le a \le b \le 5$ are relevant.  Although, by
  Lemma~\ref{lem:outer-face-length} the value of $a+b$ is fixed, this
  does usually not uniquely determine the values $a$ and $b$ in this
  case.  For example, at first sight it may seem that if the outer
  face of a pertinent graph has length 6, then uniform embeddings of
  type $(1,5)$, $(2,4)$ and $(3,3)$ may all be possible.  However, as
  we will argue in the following, only one of these choices is
  relevant in any situation.

  In order to admit a $k$-uniform embedding with $k$ even, it is
  necessary that the graph is bipartite.  In particular, this implies
  that also the outer face of any pertinent graph must have even
  length.  For a 6-uniform embedding the length of the face must be in
  $\{2,4,6,8,10\}$.  Let us now investigate for each such length the
  possible types of almost 6-uniform embeddings.

  For length~2 and length~10, the types must be $(1,1)$ and $(5,5)$,
  respectively.  For length~4, the type must be $(1,3)$ or $(2,2)$.
  However, the poles of $\skel(\mu)$ are either in the same color
  class of the bipartite graph of $G$, then only $(2,2)$ is possible,
  or they belong to different color classes, then only $(1,3)$ is
  possible.  For length~6, the possible types are $(1,5), (2,4)$ and
  $(3,3)$.  However, type $(1,5)$ implies that one of the paths
  consists of a single edge, i.e., $\mu$ is a P-node.  However, due to
  the path of length~5 on the other boundary, we need another parallel
  edge to achieve faces of size~6.  However, such an edge must be a
  Q-node child of $\mu$, showing that $(1,5)$ cannot occur.  Thus only
  $(2,4)$ and $(3,3)$ are actually possible.  Again, the color class
  of the poles determines the pair uniquely.  Finally, for length~8,
  the possible types are $(3,5)$ and $(4,4)$ and again the color
  classes uniquely determine the type.

  Thus, we know for each internal node $\mu$ precisely what must be
  the type of an almost 6-uniform embedding of $\pert(\mu)$ if one
  exists.  It remains to check whether for each node $\mu$, assuming
  that all children admit an almost 6-uniform embedding of the correct
  type, it is possible to put them together to an almost 6-uniform
  embedding of $\pert(\mu)$ of the correct type.  For this, we need to
  decide (i) an embedding of $\skel(\mu)$ and (ii) for each child
  whether to use the to mirror its almost $k$-uniform embedding.  We
  refer to the latter decision as choosing the flip of the child.

  For S-nodes, which must necessarily have length at most~6, we can
  simply try all ways to choose the flips of the children and see
  whether one of them gives the correct values.

  For a P-node, observe that, in order to obtain an almost 6-uniform
  embedding, the boundary paths of the children must be either all odd
  or all even.  If they are all even, then all pertinent graphs of
  children must have types $(2,2)$, $(2,4)$ or $(4,4)$.  Clearly, the
  children with types $(2,2)$ and $(4,4)$ have to alternate in the
  sequence, the children of type $(2,4)$ can be inserted at an
  arbitrary place.  Let $a$ and $b$ denote the number of children of
  type $(2,2)$ and $(4,4)$, respectively.  It is necessary that $|a-b|
  = 1$, otherwise they cannot alternate.  If $a > b$, then the type of
  $\pert(\mu)$ must be $(2,2)$, if $b<a$, it must be $(4,4)$ and if
  $a=b$, then it must be $(2,4)$.

  The case that all paths are odd is similar but slightly more
  complicated as there are more possible embedding types for the
  children.  The possible types are $(1,1)$, $(3,3)$, $(3,5)$, and
  $(5,5)$ (recall that $(1,3)$ and $(1,5)$ cannot occur in a P-node).
  Again, we call the corresponding virtual edges $(1,1)$-, $(3,3)$-,
  $(3,5)$- and $(5,5)$-edges, respectively.

  We now perform some simple groupings of such virtual edges that can
  be assumed to be placed consecutively in any valid embedding.  We
  view these consecutive edges as a single child whose outer boundary
  paths determine its type of embedding.  First, observe that if there
  is no $(3,5)$-edge but a $(3,3)$-edge, then all children must
  necessarily be of type $(3,3)$.  In this case any embedding of
  $\skel(\mu)$ works and yields an embedding of type $(3,3)$ for
  $\pert(\mu)$.  Otherwise, we group the $(3,5)$-edges together with
  all $(3,3)$-edges into one big chunk, which then represents a child
  of type $(3,5)$.  Thus, we can assume that no $(3,3)$-edge exists.
  Next, observe that the $(3,5)$-edges must occur in pairs whose
  interior face is bounded by two paths of length~3.  Viewed as one
  graph, the type of their embedding is $(5,5)$.  Note that, due to
  the pairing, one $(3,5)$ might be left over.  In this case, we start
  the ordering for the embedding of $\skel(\mu)$ with the
  $(3,5)$-edge.  We then alternatingly insert $(1,1)$-edges and
  $(5,5)$-edges.  If we manage to use up all virtual edges, we have
  found a valid embedding.  Otherwise, since we only followed
  necessary conditions, a valid embedding does not exist.
 
  For an R-node~$\mu$, observe that each face of the skeleton has size
  at least~3.  Thus children whose almost 6-uniform embedding has type
  $(x,5)$ for some value of $x$ immediately exclude the existence of a
  6-uniform embedding for $\pert(\mu)$.  It now remains to choose the
  flips of the almost 6-uniform embeddings of the children.  Note that
  for children whose type $(a,b)$ is such that $a=b$, this choice does
  not matter.  Thus, only the flips of children of types $(1,3)$ and
  $(2,4)$ matter.  We initially consider each face as having a demand
  of~6.  However, for each edge of type $(a,b)$ incident to a face
  $f$, we remove from the demand of face $f$ the amount $\min\{a,b\}$,
  and rather conceptually replace the edge by a $(0,|a-b|)$-edge.  Due
  to the above observation, the only types of edges remaining are
  $(0,0)$ and $(0,2)$.  Clearly, we can ignore the $(0,0)$-edges.  The
  remaining $(0,2)$-edges can pass two units of boundary length into
  one of their incident faces.  We now consider the demand of each
  face.  Clearly, it is necessary that these demands are even.  We
  then model this as a matching problem, where each $(0,2)$-edge has
  capacity~1, and each face has capacity half its demand.  We then
  seek a generalized perfect matching in the incidence graph of faces
  and vertices with positive capacity such that each vertex is matched
  to exactly as many edges as its capacity.  This can be solved in
  $O(n^{1.5})$ time by an algorithm due to Gabow~\cite{g-ertdc-83}.
  Clearly an embedding exists if and only the corresponding matching
  exists. We thus have proved the following
theorem.

\begin{theorem}
  \label{thm:6-regular-dual}
  It can be tested in $O(n^{1.5})$ time whether a biconnected planar
  graph admits a 6-uniform embedding.
\end{theorem}

\subsection{Uniform Embeddings with Large Faces}

We prove NP-hardness for testing the existence of a
$k$-uniform embedding for $k = 7$ and $k \ge 9$ by giving a reduction
from the NP-complete problem {\sc Planar Positive 1-in-3-SAT} where
each variable occurs at least twice and at most three times and each
clause has size two or three. The NP-completeness of this version of
satisfiability follows from the results of Moore and Robson~\cite{mr-htpst-01}, as shown by the following Theorem.

\newcommand{\thmpponethreesattext}{{\sc Planar Positive 1-in-3-SAT} is
  NP-complete even if each variable occurs two or three times and each
  clause has size two or three.}

  \begin{theorem}
\label{thm:pp1in3sat}
\thmpponethreesattext
  \end{theorem}

\begin{proof}
  Clearly the problem is in NP.  For the hardness proof, we reduce
  from the NP-complete problem {\sc Cubic Planar Monotone 1-in-3-SAT},
  a variant of planar 3-SAT where each variable occurs three times and
  each clause consists of three literals that are either all positive
  or all negative~\cite{mr-htpst-01}.  We denote 1-in-3 clauses as
  $(x,y,z)$ (or $(x,y)$ for clauses of size two) where $x,y,z$ are
  literals.

  Consider a planar embedding of the variable--clause graph and a
  clause $C = (\neg x, y, z)$ where a variable $x$ occurs negated.  We
  now replace $C$ by two clauses $C' = (x',y,z)$ and $C'' = (x',x)$,
  where $x'$ is a new variable.  Observe that, in the variable--clause
  graph this corresponds to subdividing the edge $xC$ twice.  Thus,
  the resulting variable--clause graph remains planar.  Further, the
  clause $C''$ ensures that, in any satisfying 1-in-3 truth
  assignment, the variables $x$ and $x'$ have complementary truth
  values, i.e., $x'$ is the negation of $x$.  Thus the resulting
  instance of {\sc Planar Positive 1-in-3-SAT} is equivalent to the
  original one.  Moreover, the new instance has one fewer negated
  literal.  After $O(n)$ such operations, we obtain an equivalent
  instance where all literals are positive.  Obviously the resulting
  formula satisfies the claimed properties and the reduction can be
  performed in polynomial time.
\end{proof}

  \begin{theorem}
    \label{thm:odd-regular-hard}
    $k$-\regdual is NP-complete for all odd $k \ge 7$.
  \end{theorem}

\begin{proof}
  We reduce from {\sc Planar Positive 1-in-3-SAT} where each variable
  occurs two or three times and each clause has size two or three,
  which is NP-complete by Theorem~\ref{thm:pp1in3sat}.  Let $\varphi$
  be such a formula with $n$ variables, $C$ clauses and $L$ literals
  (total number of literals in all clauses), and let $G_\varphi$ be
  its variable--clause graph embedded in the plane.  We add an
  additional vertex $s$, which we call \emph{sink} into the outer face
  and connect each variable to the sink in such a way that no two
  edges incident to $s$ cross.  Call this augmented graph
  $G_\varphi'$.  Note that, due to the crossings, edges may be
  subdivided into several pieces, which we call \emph{arcs}.

  In the following we will construct gadgets modeling a flow-like
  problem on $G_\varphi'$.  Each variable has $2k-1$ units of flow,
  where $d$ is the degree in $G_\varphi'$.  It sends one unit of flow
  into each incident edge.  For the remaining units of flow, it takes
  a decision.  Either it sends the remaining flow to the sink (value
  \false), or it evenly distributes it to all incident edges leading
  to a clause (value \true).  We then construct gadgets for the arcs,
  which simply pass on the information from one end to the other and
  crossing gadgets, which pass the information over crossings.  Here
  it is crucial that the crossover happens between information flows
  of different sizes.  The only crossings happen between
  variable--clause connections, which carry either one or two units of
  flow and variable--sink connections, which carry either one or
  three/four units of flow (depending on the degree of the variable).
  Since our construction is such that flows cannot be split, this
  allows to cross over these information flows.  The clauses gadgets
  are constructed such that there is a face that has size $d$ if and
  only if it receives as incoming flow the number of incident
  variables plus one, which models the fact that precisely one of them
  must be assigned the truth value \true.

  Next, we observe that for a satisfying truth assignment, there are
  precisely $C$ satisfied literals and $L-C$ unsatisfied literals in
  the formula.  Each satisfied literal ensures that only one unit is
  sent towards the sink, whereas each unsatisfied literal ensures that
  two units are sent towards the sink.  Thus, there are precisely $C +
  2(L-C) = 2L - C$ units of flow sent to $s$ via $n$ edges.  We design
  a gadget that admits an embedding where every face has size~$d$ no
  matter how the incoming flow is distributed to the edges incident to
  $s$, we call this the \emph{sink gadget}.  Let now $H_\varphi$
  denote the graph obtained from $G_\varphi$ by replacing each
  variable, arc, crossing, clause, and the sink by a corresponding
  gadget.  To ensure that the embedding of $H_\varphi$ follows the
  embedding of $G_\varphi$, we triangulate each face of $H_\varphi$
  that corresponds to a face of $G_\varphi$ and then insert into each
  triangle a construction that ensures that each of the internal faces
  has size $d$.  This fixes the planar embedding of $H_\varphi$ except
  for the decisions that are modeled by the gadgets.  It is then clear
  that $H_\varphi$ admits a planar embedding if an only if $\varphi$
  is satisfiable.

  We now give a more detailed overview of the construction.  The basic
  tool for passing information are wheels whose outer cycle has $d$
  vertices for $d=3,4,5$, and whose inner edges are subdivided
  $(k-1)/2$ times such that all inner faces have size $k$; see
  Fig.~\ref{fig:gadgets-edges}.  Note that this is possible since $k$
  is odd.  We then designate two adjacent vertices of the outer cycle
  as \emph{poles} $u$ and $v$, where it attaches to the rest of the
  graph.  The flip of this gadget then decides with of the two face
  incident to its outside is incident to a path of length~1 and which
  is incident to a path of length~$d-1$.  We call these two paths the
  \emph{boundary paths}.  In this respect, and since there internal
  faces always have size $k$, and hence are not relevant, these
  constructions behave like a single edge where one side has length~1
  and the other one has length $d-1$.  We therefore call them
  \emph{$(1,2)$-, $(1,3)$- and $(1,4)$-edges}, respectively.  We use
  them to model the flows from the above description.

  \begin{figure}[tb]
    \centering
    \begin{subfigure}[b]{.4\textwidth}
      \centering
      \includegraphics[page=1]{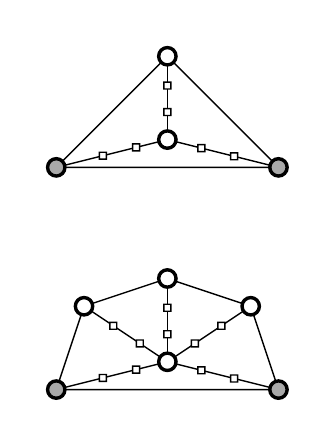}
      \caption{}
      \label{fig:gadgets-edges}
    \end{subfigure}
    \begin{subfigure}[b]{.4\textwidth}
      \centering
      \includegraphics[page=2]{fig/gadgets-short}
      \caption{}
      \label{fig:gadgets-variables}
    \end{subfigure}

    \caption{Illustration of the gadgets for the proof of
      Theorem~\ref{thm:odd-regular-hard} in the case $d=7$. (a)
      $(1,2)$-edge and $(1,4)$-edge; the poles are shaded, the
      subdivision vertices are small squares. (b) Variable gadgets for
      variables occurring two (above) and three times (below),
      respectively.  The $(1,k)$-edges are represented thick and the
      numbers give the lengths of the respective boundary paths.  Note
      that simultaneously exchanging the numbers at each edge also
      gives an embedding where the inner face has size $d$, and these
      are the only two such choices.  The edge at the bottom is the
      sink edge, the $(1,2)$-edges are the output edges.}
  \end{figure}
  
\begin{figure}[tb]
  \centering
    \begin{subfigure}[b]{.49\textwidth}
      \centering
      \includegraphics[page=3]{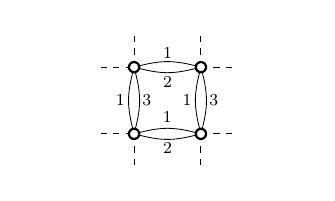}
      \caption{}
      \label{fig:gadgets-simple-crossing}
    \end{subfigure}
    \begin{subfigure}[b]{.49\textwidth}
      \centering
      \includegraphics[page=4]{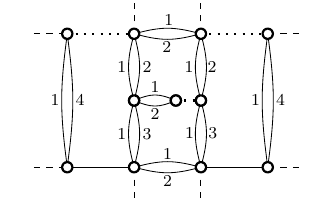}
      \caption{}
      \label{fig:gadgets-complex-crossing}
    \end{subfigure}

    \caption{Illustration of the crossing gadgets for the proof of
      Theorem~\ref{thm:odd-regular-hard} in the case $d=7$. (a) A
      crossing gadget for two edges, one carrying values in $\{1,2\}$
      and one in $\{1,3\}$.  Observe that simultaneously exchanging
      the numbers at opposite edges results again in an inner face of
      size $d$, and this can be done independently for both pairs. (b)
      A crossing gadget for two edges, one carrying values in
      $\{1,2\}$ and the other in $\{1,4\}$.  The dashed lines show
      where further gadgets attach, the length of the dotted paths
      depends on $d$, for $d=7$ their length would be~0.}
  \end{figure}
  
\begin{figure}
  \centering
   \includegraphics[page=6]{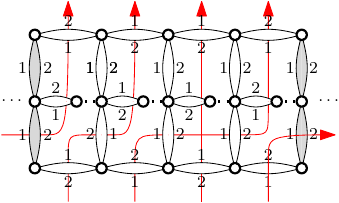}
   \caption{Illustration of a shift ring consisting of several flow
     switches.  The gray edges on the left and right boundary are
     identified.  The red arrows illustrate the flow of information,
     where the states of two $(1,2)$-edges carrying different
     information is transposed in the circular ordering compared to
     the inner face of the shift ring (above the construction) and the
     outer face of the shift ring(below the construction).}
      \label{fig:gadgets-shift-ring}
  \end{figure}

  We are now ready to describe our gadgets.  A variable of degree $d$
  in $G_\varphi'$ (recall that $d=2$ or $d=3$), the gadget is a cycle
  of length~$k-d$ such that $d$ edges are $(1,2)$-edges (the
  \emph{output edges}) and one is an $(1,d-1)$-edge (the \emph{sink
    edge}); see Fig.~\ref{fig:gadgets-variables} for variable gadgets
  for $k=7$.  Clearly, for the inner face $f$ to have size $k$, either
  all $(1,2)$-edges must be embedded such that their boundary paths of
  length~2 are incident to it and the $(1,d-1)$-edge must be embedded
  such that its boundary path of length~1 is incident to $f$, or all
  $(1,2)$-edges and the $(1,k-1)$-edge must be flipped.  This
  precisely models the flows emanated by a variable as described
  above.

  We use \emph{pipe gadgets} to transport flow along an arc.  By
  construction, each arc transports flows from exactly one of the
  three sets $\{1,2\}$, $\{1,3\}$ and $\{1,4\}$.  We give separate
  pipes for them.  Let the set of flow values be $\{1,d\}$.  The
  gadget is a cycle of length~$k-d+1$ where two nonadjacent edges are
  $(1,d)$-edges, one \emph{input} and one \emph{output} edge.
  Clearly, there are $k-d-1$ edges contributing length~1 to the inner
  face of the gadgets.  Thus, the two $(1,d)$-edges must together
  contribute paths of length~$d+1$, which occurs if and only if the
  information encoded by the input edge is transferred to the output
  edge.

  For a clause of degree $d$ in $G_\varphi$ (recall that $d=2$ or
  $d=3$), the gadget is a cycle of length~$k-1$ where $d$ edges are
  $(1,2)$-edges.  Obviously, the inner face $f$ has size $k$ if and
  only if precisely one of the $(1,2)$-edges has its boundary path of
  length~2 incident to $f$.  Thus, the gadget correctly models a
  1-in-3-SAT clause.

  For a crossing, one of the two edges transports values in $\{1,2\}$
  and the other values in $\{1,3\}$ or in $\{1,4\}$.  If it is
  $\{1,3\}$, we simply use a cycle of length~$4$, where two opposite
  edges are $(1,2)$-edges and the other two opposite edges are
  $(1,3)$-edges; see Fig.~\ref{fig:gadgets-simple-crossing}.  From
  each pair of opposite edges, we designate one as the input edge and
  one as the output edge.  It is not hard to see that the inner face
  has size $k$ if and only if the state from each input edge is
  correctly transferred to the output edge.  Of course, the same
  approach could be used for the case $\{1,4\}$, however, this would
  require $k \ge 8$.  Instead, we use a different approach; see
  Fig.~\ref{fig:gadgets-complex-crossing} for an illustration.  First,
  we split the information into two separate pieces, one that
  transmits a value in $\{1,2\}$ and one that transmits a value in
  $\{1,3\}$ (note that the sum of the differences between the upper
  and the lower values remains constant).  Then we cross over the part
  of the sink edge carrying the flow in $\{1,3\}$ as before.  To cross
  the part of the sink edge carrying flow in $\{1,2\}$ with the
  variable--clause arc, we use a gadget we call \emph{flow switch}.
  It consists of a cycle of length $k-2$ where four edges are
  $(1,2)$-edges, and for each pair of opposite $(1,2)$-edges one is
  declared the input and one is the output.  Note that, unlike the
  above crossing gadget, this does not necessarily transfer the input
  information to the correct output edge.  It only requires that half
  of the $(1,2)$-edges have a boundary path of length~2 in the inner
  face.  However, the fact that, afterwards, we use a symmetric
  construction as for splitting the flow in $\{1,4\}$ to merge the two
  flows on the sink edges after the crossing back into a flow in
  $\{1,4\}$ enforces this behavior.

  We can now construct a graph $H_\varphi'$ by replacing each variable
  by a variable gadget, each clause by a clause gadget, each crossing
  by a crossing gadget and each arc by a corresponding pipe gadget.
  The gadgets are joined to each other by identifying corresponding
  input and output edges of the gadgets (i.e., we identify the
  corresponding construction), taking into account the embedding of
  $G_\varphi'$.  For the sink, we first attach to each of the output
  edges of the pipe gadgets leading there, a corresponding variable
  gadget via its sink edge to split the flow arriving there into
  $(1,2)$-edges.  We identify the endpoints of these $(1,2)$-edges
  such that they form a simple cycle whose interior faces represents
  the sink.

  We now arbitrarily triangulate, possibly by inserting vertices, all
  faces corresponding to a face of $G_\varphi'$ that are not internal
  faces of a gadget and insert into each of the resulting triangles a
  vertex connected to each triangle vertex by a path of
  length~$(k-1)/2$.  This ensures that all resulting subfaces of the
  triangles have size $k$ and at the same time the embedding of
  $H_\varphi'$ is fixed except for the flips of the $(1,d)$-edges.
  Using the above arguments, it is not hard to see that $\varphi$
  admits a satisfying 1-in-3 truth assignment if and only if
  $H_\varphi$ admits a planar embedding where each face has size $k$
  except for the face representing the sink vertex, which is bounded
  by $L$ $(1,2)$-edges and has size $2L-C$.  To complete the proof, we
  present a construction for the interior of the sink that always
  allows an embedding where all inner faces have size $k$ as long as
  the outer face has size $2L-C$, i.e., $C$ edges have a boundary path
  of length~1 at the inner face, and the remaining $L-C$ have a
  boundary path of length~2 at the inner face.

  This works in two steps.  First, we build a \emph{shift ring}, which
  allows to shift the information encoded in a subset of the edges by
  one unit to the left or the right.  The shift ring consists of a
  ring of pairs of flow switches as shown in
  Fig.~\ref{fig:gadgets-shift-ring}.  Its inner and outer face is
  bounded by $L$ $(1,2)$-edges.  There is a natural bijection between
  the $(1,2)$-edges on the inner and on the outer ring, but the shift
  ring allows to exchange the state of two adjacent edges.  We then
  nest sufficiently many shift rings ($L^2$ certainly suffice), which
  allows us to assume that, in the innermost face, the edges whose
  boundary paths have length~2 are consecutive, and the first one (in
  clockwise direction) is at a specific position.  Second, assuming
  that the innermost face has the configuration of its $(1,2)$-edges
  as described above, we simply triangulate it arbitrarily and insert
  into each triangle the construction that makes every face have
  size~$k$.  This concludes the construction of the sink gadget, and
  thus the proof.
\end{proof}

  \begin{theorem}
\label{thm:even-regular-hard}
    $k$-\regdual is NP-complete for all even $k \ge 10$.
  \end{theorem}

\begin{proof}
  The proof runs along the lines of the proof of
  Theorem~\ref{thm:odd-regular-hard}.  For this proof, however, we
  need the additional assumption that number $L$ of literals is even.
  If it is not the case, we take a variable $x$ that occurs only twice
  (subdivide an edge to introduce such a variable as in the proof of
  Theorem~\ref{thm:pp1in3sat} if none exists).  We then create new
  variables $u,v,w$ and add the clause $(x,u,v)$ and twice the clause
  $(u,v,w)$.  If $x$ has the value \emph{true}, then setting $u = v=
  \false$ and $w = \true$ satisfies the new clauses, and if $x =
  \false$, then $u = \true$, $v=w=\false$ satisfies them.  Thus the
  resulting formula is equivalent to the original one, has an even
  number of literals, and satisfies all the conditions of
  Theorem~\ref{thm:pp1in3sat}.

  \begin{figure}[tb]
    \centering
    \includegraphics[page=7]{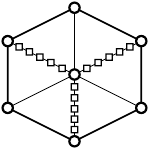}
    \caption{Construction for subdividing a face of even length (bold)
      into faces of arbitrary size $d$ (here $d=8$).}
    \label{fig:even-d-subdivide}
  \end{figure}

  Now the proof essentially reuses the construction from
  Theorem~\ref{thm:odd-regular-hard} for such a formula.  However,
  since all faces have to have even size, it is not possible to
  construct a $(1,2)$-edge (or $(1,k)$-edges with $k$ even for that
  matter); its outer face would have to have odd length while all
  interior faces have even length, which is not possible.  We thus use
  $(1,3)$-edges to transmit information for all the gadgets.  The sink
  edges can then use $(1,5)$- and $(1,7)$-edges.  A crossing gadget
  for a $(1,3)$-edge and a $(1,5)$-edge requires a face of size~10.  A
  $(1,7)$-edge can be split into a $(1,3)$ and a $(1,5)$-edge for the
  corresponding crossing gadget.  By choosing suitably long pipes, we
  can ensure that all faces that are not internal to a gadget have
  even length.  Such a face can then be subdivided into faces of size
  $d$ by adding a new vertex incident to all vertices of the face and
  subdividing every second of these edges $d-3$ times; see
  Fig~\ref{fig:even-d-subdivide}.  For the sink, we first distribute
  the information to $(1,3)$-edges using variable gadgets and then use
  corresponding shift rings made of flow switches for $(1,3)$-edges.
  Now, in the innermost face of the shift ring, there are $L$
  $(1,3)$-edges of which $C$ have length~1 in the inner face and $L-C$
  have length~3, and they can be assumed to be en bloc, starting at a
  specific edge.  The total length of the innermost face then is
  $3(L-C) + C = 3L - 2C$, which is even due to our assumption on $L$.
  Thus, the construction making every face have size $d$ can be done
  as described above.
\end{proof}

\subsubsection*{Open Problems.}  What is the complexity of $k$-\regdual
for $k=5$ and $k=8$?  Are \regdual and \minmaxface polynomial-time
solvable for biconnected series-parallel graphs?  Are they FPT with
respect to treewidth?

\paragraph{Acknowledgments.}  We thank Bartosz Walczak for discussions.
\clearpage
\bibliographystyle{abbrv}
\bibliography{regular_dual}

\begin{thebibliography}{10}

\bibitem{adp-fmep-11}
P.~Angelini, G.~{Di Battista}, and M.~Patrignani.
\newblock Finding a minimum-depth embedding of a planar graph in ${O}(n^4)$
  time.
\newblock {\em Algorithmica}, 60:890--937, 2011.

\bibitem{bm-ccvp-88}
D.~Bienstock and C.~L. Monma.
\newblock On the complexity of covering vertices by faces in a planar graph.
\newblock {\em SIAM J. Comput.}, 17(1):53--76, 1988.

\bibitem{bkrw-odfc-14}
T.~Bl{\"a}sius, M.~Krug, I.~Rutter, and D.~Wagner.
\newblock Orthogonal graph drawing with flexibility constraints.
\newblock {\em Algorithmica}, 68:859--885, 2014.

\bibitem{brw-oodc-13}
T.~Bl{\"a}sius, I.~Rutter, and D.~Wagner.
\newblock Optimal orthogonal graph drawing with convex bend costs.
\newblock In F.~V. Fomin, R.~Freivalds, M.~Kwiatkowsak, and D.~Peleg, editors,
  {\em Automata, Languages, and Programming (ICALP'13)}, volume 7965 of {\em
  LNCS}, pages 184--195. Springer, 2013.

\bibitem{dlv-sood-98}
G.~{Di Battista}, G.~Liotta, and F.~Vargiu.
\newblock Spirality and optimal orthogonal drawings.
\newblock {\em SIAM Journal on Computing}, 27(6):1764--1811, 1998.

\bibitem{dt-ogasp-90}
G.~{Di Battista} and R.~Tamassia.
\newblock On-line graph algorithms with {SPQR}-trees.
\newblock In M.~S. Paterson, editor, {\em Automata, Languages and Programming
  (ICALP'90)}, volume 443 of {\em LNCS}, pages 598--611. Springer, 1990.

\bibitem{fkmp-cimrp-95}
M.~R. Fellows, J.~Kratochv{\'i}l, M.~Middendorf, and F.~Pfeiffer.
\newblock The complexity of induced minors and related problems.
\newblock {\em Algorithmica}, 13:266--282, 1995.

\bibitem{g-ertdc-83}
H.~N. Gabow.
\newblock An efficient reduction technique for degree-constrained subgraph and
  bidirected network flow problems.
\newblock In {\em Theory of Computing (STOC'83)}, pages 448--456. ACM, 1983.

\bibitem{gt-ccurpt-01}
A.~Garg and R.~Tamassia.
\newblock On the computational complexity of upward and rectilinear planarity
  testing.
\newblock {\em SIAM J. on Comput.}, 31(2):601--625, 2001.

\bibitem{gm-lis-01}
C.~Gutwenger and P.~Mutzel.
\newblock A linear time implementation of {SPQR}-trees.
\newblock In J.~Marks, editor, {\em Graph Drawing (GD'00)}, volume 1984 of {\em
  LNCS}, pages 77--90. Springer, 2001.

\bibitem{gm-emmeea-03}
C.~Gutwenger and P.~Mutzel.
\newblock Graph embedding with minimum depth and maximum external face
  (extended abstract).
\newblock In G.~Liotta, editor, {\em Graph Drawing (GD'03)}, volume 2912 of
  {\em LNCS}, pages 259--272. Springer, 2004.

\bibitem{mr-htpst-01}
C.~Moore and J.~M. Robson.
\newblock Hard tiling problems with simple tiles.
\newblock {\em Discrete Comput. Geom.}, 26(4):573--590, 2001.

\bibitem{mw-ocep-99}
P.~Mutzel and R.~Weiskircher.
\newblock Optimizing over all combinatorial embeddings of a planar graph
  (extended abstract).
\newblock In G.~Cornu{\'e}jols, R.~E. Burkard, and G.~J. Woeginger, editors,
  {\em Integer Programming and Combinatorial Optimization (IPCO'99)}, volume
  1610 of {\em LNCS}, pages 361--376. Springer, 1999.

\bibitem{w-epgmnlc-02}
G.~J. Woeginger.
\newblock Embeddings of planar graphs that minimize the number of long-face
  cycles.
\newblock {\em Oper. Res. Lett.}, pages 167--168, 2002.

\end{thebibliography}

%
%
%
%
%
%
%
%
%
%
%
%
%
%
%
%
%
%
%
%
%
%
%
%
%
%
%
%
%
%
%
%
%
%
%
%

\end{document}